\documentclass[12 pt]{amsart}
\pdfoutput=1
\usepackage[numbers,sort]{natbib}
\usepackage{amssymb,amsfonts,amsthm, color}
\usepackage{graphicx}
\usepackage{bm}
\newtheorem{thm}{Theorem}[section]

\newtheorem{prop}[thm]{Proposition}

\theoremstyle{definition}

\theoremstyle{remark}

\numberwithin{equation}{section}

\begin{document}
\title[Thermal conductivity]{Thermal conductivity and local thermodynamic equilibrium
of stochastic energy exchange models}
\author{Yao Li}
\address{Yao Li: Department of Mathematics and Statistics, University
  of Massachusetts Amherst, Amherst, MA, 01002, USA}
\email{yaoli@math.umass.edu}

\author{Wenbo Xie}
\address{Wenbo Xie: Department of Mathematics and Statistics, University
  of Massachusetts Amherst, Amherst, MA, 01002, USA}
\email{wenboxie@umass.edu}
\thanks{Wenbo Xie was partially supported by the REU program at University
  of Massachusetts Amherst}


\begin{abstract}
  In this paper we study macroscopic thermodynamic properties of a
  stochastic microscopic heat conduction model that is reduced from
  deterministic problems. Our goal is to numerically check how the
  ``low energy site effect'' inherited from the deterministic model
  would affect the macroscopic thermodynamic properties such as the
  thermal conductivity and the local thermodynamic equilibrium. After
  a series of numerical computations, our
  conclusion is that neither the thermal conductivity nor the
  existence of local thermodynamic equilibrium is qualitatively
  changed by this effect. 
\end{abstract}
\maketitle

\section{Introduction}
In general, nonequilibrium statistical mechanics is not as
well-developed as its equilibrium counterpart. Mathematical
justifications to many fundamental problems in nonequilibrium
statistical physics are not complete yet. The derivation of Fourier's
law from microscopic Hamiltonian dynamics is one of such century-old
challenge. It is not clear yet how macroscopic thermodynamic laws
including Fourier's law can be rigorously proved from the motion and
interactions of a large number of Newtonian particles \cite{bonetto2000fourier}. 

A more precise example is a long and thin tube that contains many kinetic
particles. A particle only does free motion and elastic
collisions. Now assume two ends of this tube is thermalized in a way
that the particle collides with a random particle chosen from a
Boltzmann distribution when hitting the left or right boundary. When the
temperatures of these two Boltzmann distributions are distinct, the
system is driven out from its thermal equilibrium by the boundary
effect. Needless to say, this problem is far beyond the reach of
today's dynamical systems technique. In fact, most results about
dynamical billiards are for one-particle billiard systems
\cite{chernov2005billiards}, with only a few exceptions
\cite{simanyi1999hard, simanyi2003proof}.

In our earlier paper \cite{li2015chaos}, we attempted to reduce this
``particle in a tube'' problem to a mathematically tractable
stochastic energy exchange model by numerical simulations. The idea is
to divide the tube into a large number of localized cells as in
\cite{bunimovich1992ergodic}, such that each particle is trapped in a
cells, but collisions between particles in adjacent cells are still
allowed through the opening between neighboring cells. See Figure \ref{table} for the detail. Then we use
numerical tools to investigate the rule of energy exchanges between
cells. We refer Section 2.2 for a detailed description of the energy
exchange rule. This gives a stochastic energy exchange
model that approximates the time evolution of the energy profile of
the billiards model. The stochastic energy exchange model consists of
a chain of sites that is connected to two heat baths at its
ends. Each site carries some energy, which can be exchanged with
neighboring sites at exponentially distributed random times. The rule
of energy redistribution 
at an energy exchange is also random. Many rigorous results can be
proved for the resultant 
stochastic energy exchange model. Among which, our earlier paper
\cite{li2016polynomial} rigorously proved that the speed of convergence to the
steady state, i.e., the nonequilibrium steady state (NESS), of this model is polynomial. On the other hand, a slightly different stochastic
exchange model can be derived by working on the time rescaling limit
when particles in adjacent cells barely collide \cite{gaspard2008heat,
gaspard2008heat2, gaspard2008derivation}. It is
known that the speed of convergence of the second model is exponential
\cite{grigo2012mixing, li2013existence, sasada2015spectral}. 

\begin{figure}[h]
\centerline{\includegraphics[width = \linewidth]{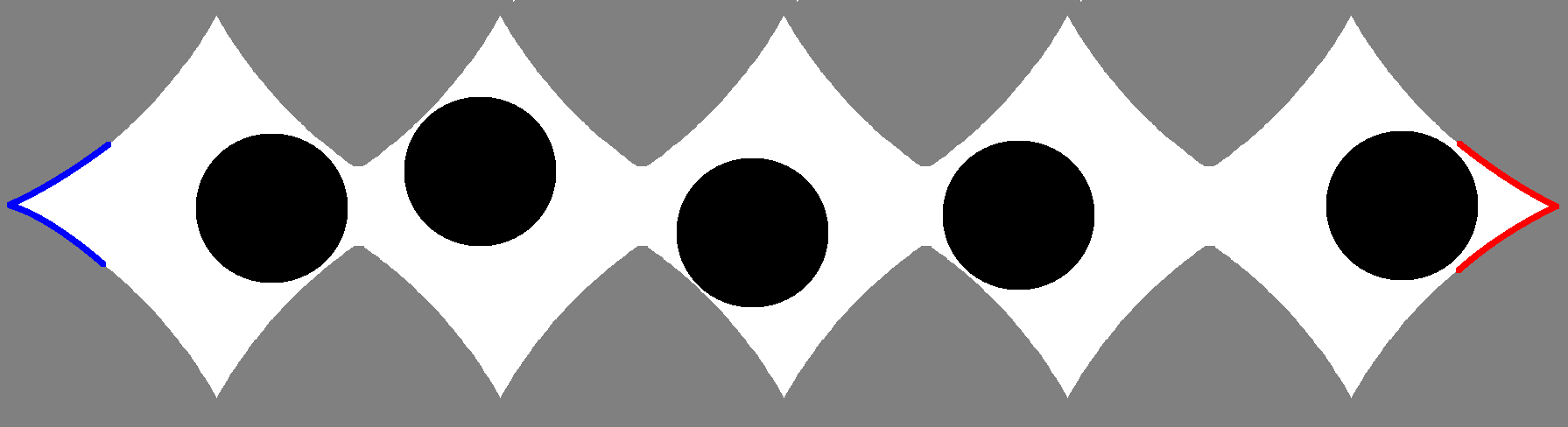}}
\caption{A billiards-like microscopic heat conduction model. Each particle
  is trapped in its own cell. Collisions through the opening between
  twe cells are allowed. Red and
  blue sections of the boundary is thermalized in a why that the
  particle receives a random kick after colliding with these two
  sections. The random kick mimics a collision with a particle drawn
  from a Boltzmann distribution. }
\label{table}
\end{figure}

The slow speed of convergence of the model in \cite{li2016polynomial} is due to
the presence of low energy 
particle. Because of the localization, the next energy
exchange will not happen in a long time period if the kinetic energy
of one of the involved particle is low. The slow particle has to move to the ``gate'' by
itself in order to exchange energy with others. As a result, the energy transport is temporarily
blocked by this low energy particle. The stochastic energy exchange model
inherits this feature from the original deterministic heat conduction
model. If a site carries a very low amount of energy, it will wait a
long time for the next energy exchange. We call this the {\it low
  energy site effect}. Since the energy transport is
occasionally halted by low energy sites, one natural question is that:
Would the low energy site effect in a stochastic energy 
exchange model qualitatively changes the thermal conductivity?

A more fundamental question is about the existence of the local thermodynamic
equilibrium (LTE). The existence of LTE means that the marginal
distribution of the nonequilibrium steady state with respect to finite
local
sites converges to a thermal equilibrium when the 
length of the chain approaches to infinity. Heuristically, this
implies the existence of a well-defined local temperature. There are
very limited rigorous results about the existence of LTE due to its
significant difficulty \cite{kipnis1982heat, ravishankar2007local, li2016local}, all of which are for very simple
heat conduction models. It is also tempting to check, whether the low
energy site effect would make the stochastic energy exchange model fail to achieve LTE. 

Different from the thermal equilibrium, NESS usually does
not have an explicit form. We are able to prove its existence, uniqueness,
ergodicity, and hydrodynamic limits in some situations. But in general
a detailed description of NESS is not possible. In fact, this is one
reason why any rigorous justification of nonequilibrium statistical
physics is challenging. Since mathematical studies to the
thermal conductivity and LTE are too difficult, we will have to
seek help from numerical simulations. 

The main subject of this paper is to answer the two questions raised above
numerically. In a stochastic energy exchange model, we choose two
different rate functions corresponding to exponential
and polynomial ergodicity, respectively. The we use law of large
numbers of martingale difference sequences to show that the thermal
conductivity is both well defined and computable through Monte Carlo
simulations. And the marginal distribution of the NESS is obviously
well defined and computable because of the ergodicity. Hence it is
not difficult to design a series of numerical simulations to compute
the thermal conductivity and the marginal distribution. Our
simulations are implemented by the Hashing-Leaping Method (HLM) developed in
\cite{li2015fast}, which is significantly faster than most
implementation methods of the stochastic
simulation algorithm (SSA). Parallel computing is used to collect enough samples.

Our numerical simulations shows that the low energy site effect
will not qualitatively affect the thermal conductance, which is
supposed to be proportional to the reciprocal of the length of the
chain. This implies the existence of a ``normal'' thermal
conductivity. The thermal conductivity of the model with slow speed of
convergence can be increased by changing to 2D. Then the effect of
low energy site is significantly reduced. The existence of LTE is a
more subtle issue. To check it, one needs to accurately compute the marginal distribution of the NESS. However, the slow
convergence speed to NESS caused by the low energy effect imposes many
challenges to such computation. After working carefully on the
sampling technique and the algorithm, we conclude that LTE is achieved
in our model regardless affected by the low energy effect or not.

The paper is organized in the following way. The stochastic energy
exchange model, its connection to deterministic dynamical system, and
relevant rigorous results are introduced in Section 2. Section 3 is
about the law of large number and numerical results of the thermal
conductivity. The existence of LTE is investigated in Section
4. Section 5 is the conclusion.

\section{Model Description}
\subsection{Reduction from deterministic dynamics}
Consider an 1D chain of billiard tables as described in Figure \ref{table}, called
the locally confined particle system. One
disk-shaped particle is ``trapped'' in a billiard table such that each
particle is allowed to collide with those particles in adjacent
billiard tables but can not leave its billiard table. In addition, we assume that the boundary of each
billiard table is piecewise $C^{3}$ and strictly convex inward, so that it
forms a chaotic dynamical billiards system by itself \cite{chernov2006chaotic}. This model is
intensively studied because this is probably the simplest
deterministic dynamical system that models the microscopic heat
conduction. The kinetic energy is transported through collisions
between particles.

Due to the significant difficulty of studying a chaotic multibody
system, a natural question is that whether one can reduce this
deterministic dynamical system to a Markov process. More precisely, we
look for a stochastic energy exchange process that only keep track of
the time evolution of the energy profile. Obviously the process of
energy evolution is not Markovian. But since a chaotic billiards
system has very good statistical properties, we expect this
deterministic energy evolution process to be well approximated by a
Markov process, at least under some rescaling limit. 

There has been two different studies about the reduction from the
billiard system in Figure \ref{table} to a Markov
process. One study was conducted by \cite{gaspard2008heat, gaspard2008heat2}, which essentially assumes
that the gap between two tables is extremely small. Then we can take a
time rescaling limit such that the expected number of
particle-particle collisions per unit time is still $1$. The
conclusion of this study is that at this time rescaling limit, the
probability that two particles with energy $(E_{1}, E_{2})$ collides
during the next time interval with length $\mathrm{d}t \ll 1$ is
approximately $\sqrt{E_{1} + E_{2}} \mathrm{d}t$. Now assume the
energy exchange process is Markov. Then the interval between two
consecutive energy exchanges should be an exponentially distributed
random variable, whose rate is $\sim \sqrt{E_{1} + E_{2}}$. We refer readers to
\cite{gaspard2008heat2} for the precise formula of the energy exchange kernel.

The other point of view, however, focuses on the dynamics at the {\it
  original} time scale. If the billiards table is properly chosen, the
time distribution of the next particle-particle collision is very
close to an exponential distribution. Instead of taking the time rescaling limit,
one can numerically probe the slope of the exponential tail of the
first collision time. Additional simulations in \cite{li2015chaos} demonstrate that the
conditional distribution of the time duration between two consecutive
collisions have the same exponential tail. Therefore, the energy
exchange times of the billiards model can be approximated by a
Poisson clock. The rate of this clock, or the slope of the exponential
tail, is called the {\it stochastic energy
  exchange rate}. When two adjacent particles have energies
$(E_{1}, E_{2})$, the numerical simulation in \cite{li2015chaos} shows that
the slope of this exponental tail is $\sim \sqrt{ \min \{E_{1}, E_{2}\}}$. In
other words, the rate of the exponential clock about the energy
exchange event should be $\sim \sqrt{ \min \{E_{1}, E_{2}\}}$. This
rate respects the dynamics of the billiards system at its original
time scale. It is easy to see that a slow particle needs a long time
to move to the ``gate area'' in order to have a collision, which causes
the low energy site effect. Hence the
next collision time mainly depends on the lower particle energy in a
nearest neighbor pair particles. We refer \cite{li2015chaos} for further
discussions about this clock rate. 

It remains to discuss the rule of energy redistribution at a
collision. The analysis and numerical simulation in \cite{li2018billiards} shows
that although the explicit formula of an energy redistribution is
too complicated to be useful, the amount of
exchanged energy has positive density everywhere. Hence it is proper
to assume that the energy repartition is done in a ``random halves''
way as described in equation \eqref{update}. More precisely, we assume that the energies of two colliding
particles are pooled together at first. Then a (uniformly distributed) random
proportion of the total energy goes to the left, and the rest energy
goes to the right. This simplified rule has been used in many early
studies \cite{grigo2012mixing, kipnis1982heat, sasada2015spectral,
  li2013existence}.

\subsection{Stochastic energy exchange process}

In summary, the locally confined particle system in Figure \ref{table}
can be reduced to the
following two stochastic energy processes with two different rate
functions. Each process corresponds to one approach of model
reduction. Consider a chain of $N$ sites
carrying energy $E_{1}, \cdots, E_{N}$ respectively. An exponential clock is
associated to each pair of sites $(E_{i}, E_{i+1})$. The rate of this
clock is $R(E_{i}, E_{i+1})$. When the clock rings, an energy exchange
event occurs immediately. The rule of energy exchange is that
\begin{equation}
  \label{update}
(E'_{i}, E'_{i+1}) = (p(E_{i} + E_{i+1}), (1-p)(E_{i} + E_{i+1}) )\,,
\end{equation}
where $p$ is a uniform random variable on $(0, 1)$. We assume the rate function $R(E_{i}, E_{i+1})$ has two different
choices $R = R_{1}(E_{i}, E_{i+1}) = \sqrt{E_{i} + E_{i+1}}$ and $R =
R_{2}(E_{i}, E_{i+1}) = \sqrt{E_{i}E_{i+1}/(E_{i} + E_{i+1})}$,
corresponding to the dynamics at the time rescaling limit and the
original time scale, respectively. Note that here we choose $R_{2}(E_{i},
E_{i+1}) = \sqrt{E_{i}E_{i+1}/(E_{i} + E_{i+1})}$ because it is a smooth function that
mimics the shape of $\sqrt{\min \{ E_{i}, E_{i+1}\}}$, and it admits an
explicit thermal equilibrium.

The rule of energy exchange with the boundary is the same. We assume
this chain is connected to two heat baths with temperatures $T_{L}$
and $T_{R}$ respectively. Two more exponential clocks with rates
$R(T_{L}, E_{1})$ and $R(E_{N}, T_{R})$ are associated to two ends of
the chain respectively. When the left (resp. right) clock rings, the
first (resp. last) site updates energy according to the following rule
\begin{equation}
\label{updatebc}
  E_{1}' = p(E_{1} + \mathcal{E}(T_{L})) \quad (\mbox{resp. } E_{N}' =
  p(E_{N} + \mathcal{E}(T_{R})  ))\,,
\end{equation}
where $p$ is a uniform random variable on $(0, 1)$, and
$\mathcal{E}(\lambda)$ mean an exponential random variable with mean
$\lambda$. 

The stochastic energy exchange process described above generates a
Markov process $\Phi_{t}$ on $\mathbb{R}^{N}_{+}$. Let $\mathbf{E} = (E_{1},
\cdots, E_{N})\in
\mathbb{R}^{N}_{+}$ be a state of the Markov process and
$f(\mathbf{E})$ be a measurable function. To distinguish the two rate
functions, we denote the Markov process by $\Phi^{1}_{t}$ if the rate
function is $R_{1}$ and by $\Phi^{2}_{t}$ if the rate function is
$R_{2}$. The upper index is dropped when it does not lead to a confusion.

The infinitesimal generator $\mathcal{L}_{i}$ of $\Phi^{i}_{t}$ for $i = 1, 2$ is 
\begin{align}
\label{generator}
\mathcal{L}_{i} f(\mathbf{E}) & =  \sum_{n = 1}^{N-1}R_{i}(E_{n},
                                    E_{n+1})[\int_{0}^{1}f(E_{1}, \cdots, p
                                    (E_{n} + E_{n+1}), (1-p)(E_{n} +
                                    E_{n+1}), \\\nonumber
&\cdots, E_{N} )\mathrm{d}p  - f(\mathbf{E})  ]\\\nonumber
& + R_{i}(T_{L}, E_{1}) [\int_{0}^{\infty}\int_{0}^{p} f( p(E_{1} +
   x), E_{2}, \cdots, E_{N}) \frac{1}{T_{L}}e^{-x/T_{L}} \mathrm{d}x
   \mathrm{d}p - f(\mathbf{E})] \\\nonumber
& + R_{i}(E_{N}, T_{R}) [\int_{0}^{\infty}\int_{0}^{p} f( E_{1},
   \cdots, E_{N - 1}, p(E_{N} + x)) \frac{1}{T_{R}}e^{-x/T_{R}} \mathrm{d}x
   \mathrm{d}p - f(\mathbf{E})] \,.\nonumber
\end{align}

\subsection{Rigorous results for the stochastic energy exchange process}

Let $V(\mathbf{E})$ be a strictly positive function. For any
signed measure $\mu$ on $\mathbb{R}^{N}_{+}$, denote
\begin{equation}
\label{norm}
  \|\mu\|_{V} = \int_{\mathbb{R}^{N}_{+}} V(\mathbf{E})|\mu|(
  \mathrm{d} \mathbf{E})
\end{equation}
by the $V$-weighted total variation norm and $\|\mu\|_{TV}$ by the
total variation norm. Further, let $L_{V}( \mathbb{R}^{N}_{+})$ be the
collection of $V$-integrable probability measures.

We have the following results for $\Phi^{1}_{t}$ (from \cite{li2013existence})
and $\Phi^{2}_{t}$ (from \cite{li2016polynomial}). 

\begin{thm}
\label{exp}
$\Phi^{1}_{t}$ admits a unique invariant probability measure that is
absolutely continuous with respect to the Lebesgue measure. In
addition, there exist constants $c > 0$ and $ \rho \in (0, 1)$ such that
$$
  \|P^{t}( \mathbf{E}, \cdot ) - \pi \|_{V} \leq c V( \mathbf{E}) \rho^{t}
$$
for every $\mathbf{E} \in \mathbb{R}^{N}_{+}$, where $V( \mathbf{E}) =
1 + \sum_{i = 1}^{N} E_{i}$.
\end{thm}

\begin{thm}
\label{poly}
Assume further that there exists a constant $K \gg T_{L}, T_{R}$ such
that $R(E_{i}, E_{i+1}) = \min\{K, \sqrt{\min\{E_{i},
  E_{i+1}}\}$. Then $\Phi_{t}$ admits a unique invariant probability measure that is
absolutely continuous with respect to the Lebesgue measure. In
addition, for any $\gamma > 0$, there exists $\eta > 0$ such that for
any $\mu \in L_{V_{\eta}}( \mathbb{R}^{N}_{+})$, 
$$
  \lim_{t \rightarrow \infty} t^{1 - \gamma} \| \mu P^{t} - \pi
  \|_{TV} = 0 \,,
$$
where
\begin{equation}
\label{lyapunov}
  V_{\eta} = \sum_{i = 1}^{N} E_{i} + \sum_{m = 1}^{N} \sum_{i = 1}^{N
  - m + 1}(\sum_{j = 0}^{m - 1} E_{i+j})^{a_{m} \eta - 1} \,,
\end{equation}
and $a_{m} = 1 - (2^{m-1} - 1)/(2^{N} - 1)$ for $m = 1, \cdots, N$. 
\end{thm}

We remark that a slightly different rate function $R(E_{i}, E_{i+1}) = \min\{K, \sqrt{\min\{E_{i},
  E_{i+1}}\}$ is used in Theorem \ref{poly} for technical reasons in
order to make a rigorous proof possible in \cite{li2016polynomial}. It
has the same scaling as $R_{2}$ in low energy configurations but makes the proof
much simpler (which still contains $35$ pages technical
calculation). We expect the speed of convergence to the invariant
probability measure of
$\Phi^{2}_{t}$ to be the same as described in Theorem \ref{poly}. In
other words, the ergodicity of $\Phi^{1}_{t}$ and $\Phi^{2}_{t}$ are qualitatively different. The speed of convergence to the steady
state is exponential for $\Phi^{1}_{t}$ but polynomial for
$\Phi^{2}_{t}$.

\section{Comparison of thermal conductivity}
As discussed in the previous section, two rate functions generate two
Markov processes $\Phi^{1}_{t}$ and $\Phi^{2}_{t}$ with very different
asymptotic properties. $\Phi^{2}_{t}$ has a much slower speed of convergence to
its invariant probability measure due to the low energy site effect, which is inherited
from the deterministic billiard model. As a result, after an energy
exchange event of $\Phi^{2}_{t}$, if a site
gets a very low amount of energy, the energy transport will be blocked
for a while until this low energy site ``recovers'' by itself. One natural question is that: how much would the
 low energy site effect affect macroscopic thermodynamic properties?
Would it cause an ``abnormal'' thermal
conductivity that depends on the system size? In this section, we will address this
issue numerically.

\subsection{Thermal conductivity for 1D model}
Let $\pi$ be the invariant probability measure of $\Phi_{t}$. The {\it
  thermal conductivity} of the stochastic energy exchange model is defined as
\begin{align}
\label{kappa1}
  \kappa &=  \frac{1}{T_{R} - T_{L}}\int \left \{ (\sum_{ i = 1}^{N-1}
  R(E_{i}, E_{i+1})\int_{0}^{1}p(E_{i} + E_{i+1}) \mathrm{d}p - E_{i} ) \right .\\\nonumber
&+ R(T_{L}, E_{1})(E_{1} -
  \int_{0}^{\infty} \int_{0}^{1}p(E_{1} + x)e^{-x/T_{L}} \mathrm{d}p
  \mathrm{d}x )\\\nonumber
& \left .+ R(E_{N}, T_{R})(  \int_{0}^{\infty} \int_{0}^{1}p(E_{N} + x)e^{-x/T_{R}} \mathrm{d}p
  \mathrm{d}x - E_{N})\right \}
  \pi( \mathrm{d} \mathbf{E}) \,.
\end{align}
Equation \eqref{kappa1} reflects the ratio of the energy flux to
the temperature gradient within an
infinitesimal amount of time when starting from $\pi$. We claim that
$\kappa$ is a computable quantity, i.e., the law of large
numbers can be applied to $\kappa$. The {\it thermal
  conductance}, denoted by $\mathbf{q}$, is the ratio of $\kappa$ to
the system size, i.e., $\mathbf{q} = \kappa/(N+1)$.

Let $t_{0}< t_{1} < t_{2} <\cdots$ be the time at which an energy
exchange occurs. Let $J_{i}$ be the energy flux from right to left
associated to the energy exchange event occurring at time $t_{i}$. If
the energy exchange event is between site $k$ and site $k+1$, we have
$J_{i} = E_{k}(t_{i}^{+}) - E_{k}(t_{i})$. If the energy exchange is
between site $1$ (resp. site $N$) and the left (resp. right) boundary,
we have $J_{i} = E_{1}(t_{i}) - E_{1}(t_{i}^{+})$ (resp. $J_{i} =
E_{N}(t_{i}^{+}) - E_{1}(t_{i})$). 

\begin{thm}
\label{lln}
Assume there exists a constant $K \gg T_{L}, T_{R}$ such that $R$ can
not exceed $K$. In other words, $R_{1}$ and $R_{2}$ are modified to
$\min\{K, \sqrt{E_{i} + E_{i+1}}\}$ and $\min\{K,
\sqrt{\frac{E_{i}E_{i+1}}{E_{i} + E_{i+1}}} \}$ respectively. 
Assume further $\pi(| \mathbf{E}|^{2}) < \infty$, then
\begin{equation}
\label{kappa2}
  \kappa = \lim_{T \rightarrow \infty}\frac{1}{T}
  \frac{1}{T_{R} - T_{L}} \sum_{t_{i} < T} J_{i} < \infty \quad a.s.
\end{equation}
\end{thm}

\medskip
We remark that two ``assumptions'' in Theorem \ref{lln} are actually
provable with extra work. Since the theme of the present paper is
about numerical computations, we simply assume these properties to avoid further
distractions. A closer look to the proofs in paper
\cite{li2013existence} and \cite{li2016polynomial} reveals that
$\pi(\| \mathbf{E}\|_{1}^{2}) $ is finite for both $\Phi^{1}_{t}$ and
$\Phi^{2}_{t}$. And the assumption of an upper bound $K$ can
be removed by using the estimation of expected energy gain introduced
in Proposition 5.1 of \cite{li2013existence}.

\medskip

\begin{proof}
Let $0 < h \ll 1$ be a time step. Let $\Phi_{n} := \Phi_{hn}$ be the
time-$h$ sample chain of $\Phi_{t}$. Let $Y_{n}$ be the total energy
flux during the time period $[nh, (n+1)h)$, i.e.,
\begin{equation}
\label{eq1}
  Y_{n} = \frac{1}{h}\frac{1}{T_{R} - T_{L}}\sum_{nh \leq t_{i} <
    (n+1)h} J_{i} \,.
\end{equation}
Then we have
\begin{equation}
\label{eq2}
  \lim_{T \rightarrow \infty}\frac{1}{T}
  \frac{1}{T_{R} - T_{L}} \sum_{t_{i} < T} J_{i} = \lim_{m \rightarrow
  \infty } \frac{1}{m}\sum_{n
    = 1}^{m}Y_{n} \,.
\end{equation}
Hence it is sufficient to prove the law of large numbers for $Y_{n}$. 

Let $Z_{n} = \mathbb{E}_{\Phi_{n}}[Y_{n}]$. It is easy to see that
$Z_{n}$ is an observable of $\Phi_{n}$. Now let $\mathcal{F}_{n}$ be
the $\sigma$-field generated by $\Phi_{0}, \cdots, \Phi_{n}$. Let
$Y'_{n} = Y_{n} - Z_{n}$. It is easy to see that 
$$
  \mathbb{E}[Y'_{n} \,|\, \mathcal{F}_{n}] = 0 \,.
$$
Hence $Y'_{n}$ is a martingale difference sequence with respect to
$\{\mathcal{F}_{n}\}_{n \geq 0}$. It is well known that if 
\begin{equation}
  \label{2ndmoment}
  \sum_{n = 1}^{\infty} \frac{\mathbb{E}[|Y'_{n}|^{2}]}{n^{2}} <
  \infty \,,
\end{equation}
we have
$$
  \lim_{m \rightarrow \infty}\frac{1}{m}\sum_{n = 1}^{m} Y'_{n} = 0 \,.
$$
(This is the law of large numbers for martingales, see for example
Theorem 3.3.1 of \cite{fazekas2001general}.)
In addition, we have
\begin{equation}
\label{eq3}
  \mathbb{E}[|Y'_{n}|^{2}] \leq \mathbb{E}[|Y_{n}|^{2}] \leq
  (\frac{1}{h}\frac{1}{T_{R} - T_{L}})^{2}
\mathbb{E}[ \mathbb{E}_{\Phi_{n}}[  \sum_{nh \leq t_{i} < (n+1)h} \|
  \Phi_{t_{i}} + X_{i}\|^{2}_{1}] ]\,, 
\end{equation}
where $X_{i}$ are i.i.d. random variables with law $\max\{\mathcal{E}(T_{L}),
\mathcal{E}(T_{R})\}$. This is because $J_{i}$ can not exceed the sum
of total energy and the energy coming from the boundary. In addition,
assume $t_{i}, t_{i+1}, \cdots t_{i+m}$ are the first $m+1$ energy exchange times right after $nh$, then the
update of the total energy is bounded by $\|\Phi_{t_{i+m}^{+}}\|_{1}
\leq \|\Phi_{n} \|_{1} + X_{i} + \cdots + X_{m}$. 

Since the clock rate is bounded by $K$ from above, it is easy to see
that
\begin{equation}
\label{eq4}
  \mathbb{E}_{\Phi_{n}}[  \sum_{nh \leq t_{i} < (n+1)h} \|
  \Phi_{t_{i}} \|^{2}_{1}] \leq \mathbb{E}[\sum_{i = 1}^{\mathbf{N}}(
  \|\Phi_{n}\|_{1} + 2 X_{i} )^{2}] \,,
\end{equation}
where $\mathbf{N}$ is a Poisson random variable with mean $Kh$. Then
some straightforward calculation shows that there exist an $h_{0} >
0$, such that for any $h$ we have
\begin{equation}
\label{eq5}
  \mathbb{E}[\sum_{i = 1}^{\mathbf{N}}(
  \|\Phi_{n}\|_{1} + 2 X_{i} )^{2}] \leq C_{0} \max\{
  \|\Phi_{n}\|_{1}^{2}, 1\}\,,
\end{equation}
where the constant $C_{0}$ only depends on $h_{0}$. 

By the law of large number of Markov process, we have
\begin{equation}
\label{eq6}
  \mathbb{E}[(\max\{ \| \Phi_{n}\|_{1}, 1\} )^{2}] \rightarrow
  \pi((\max\{ \| E \|_{1}, 1\} )^{2}) < \pi(\| E \|_{1}^{2} ) + 1 < \infty \,.
\end{equation}
Then by equations \eqref{eq3}-\eqref{eq6}, we have
\begin{equation}
\label{eq7}
  \sum_{n = 1}^{\infty} \frac{\mathbb{E}[|Y'_{n}|^{2}]}{n^{2}} < \infty \,.
\end{equation}
Hence the condition in equation \eqref{2ndmoment} is satisfied, and
the law of large numbers of $Y'_{n}$ holds. We have
$$
  \lim_{m \rightarrow \infty} \sum_{n = 1}^{m} Y'_{n} = 0 \,.
$$

In addition, by the
law of large numbers of Markov process, we have
\begin{equation}
\label{eq71}
  \frac{1}{m}\sum_{n = 1}^{m} Z_{n} \rightarrow
  \frac{1}{h}\mathbb{E}_{\pi}[\sum_{t_{i} < h} J_{i}] 
\end{equation}
almost surely. And the right hand side of equation \eqref{eq71} is
finite because clock rates can not exceed $K < \infty$. 
Therefore, we have
\begin{equation}
  \label{eq8}
  \lim_{T \rightarrow \infty}\frac{1}{T}
  \frac{1}{T_{R} - T_{L}} \sum_{t_{i} < T} J_{i} = \lim_{m \rightarrow
  \infty} \frac{1}{m}\sum_{n = 1}^{m}Y_{n} = \frac{1}{m}\sum_{n =
  1}^{m} Z_{n} + \lim_{m \rightarrow \infty} \sum_{n = 1}^{m} Y'_{n} = 
  \frac{1}{h}\mathbb{E}_{\pi}[\sum_{t_{i} < h} J_{i}] < \infty
\end{equation}
almost surely.

By the invariance of $\pi$, the quantity
$\frac{1}{h}\mathbb{E}_{\pi}[\sum_{t_{i} < h} J_{i}]$ is independent
of $h$. Since $h$ can be arbitrarily small, by equation \eqref{eq8}, we have
\begin{equation}
  \label{eq9}
   \lim_{T \rightarrow \infty}\frac{1}{T}
  \frac{1}{T_{R} - T_{L}} \sum_{t_{i} < T} J_{i} =
  \lim_{h \rightarrow 0} \frac{1}{h}\mathbb{E}_{\pi}[\sum_{t_{i} < h} J_{i}] = \kappa \,,
\end{equation}
where $\kappa$ is the infinitesimal flux defined in equation \eqref{kappa1}.
\end{proof}

Then we compute thermal conductivities of $\Phi^{1}_{t}$ and
$\Phi^{2}_{t}$ by simulating $\kappa$. 

\medskip

{\bf Numerical Simulation 1:} 
Let $T_{L} = 1$, $T_{R} = 2$. According to Theorem \ref{lln}, we can
compute the thermal conductance
$$
 \mathbf{q} = \lim_{T \rightarrow \infty}\frac{1}{T}\frac{1}{N+1}
  \frac{1}{T_{R} - T_{L}} \sum_{t_{i} < T} J_{i}
$$
over a long trajectory. In our simulation $T$ is chosen to be $2
\times 10^{6}$. We compute the thermal conductance $\mathbf{q}$ for $N = 6, 8, 10, \cdots,
100$. The simulation results for $\Phi^{1}_{t}$ and $\Phi^{2}_{t}$ are
presented in Figure \ref{fig1}.

\begin{figure}[h]
\centerline{\includegraphics[width = \linewidth]{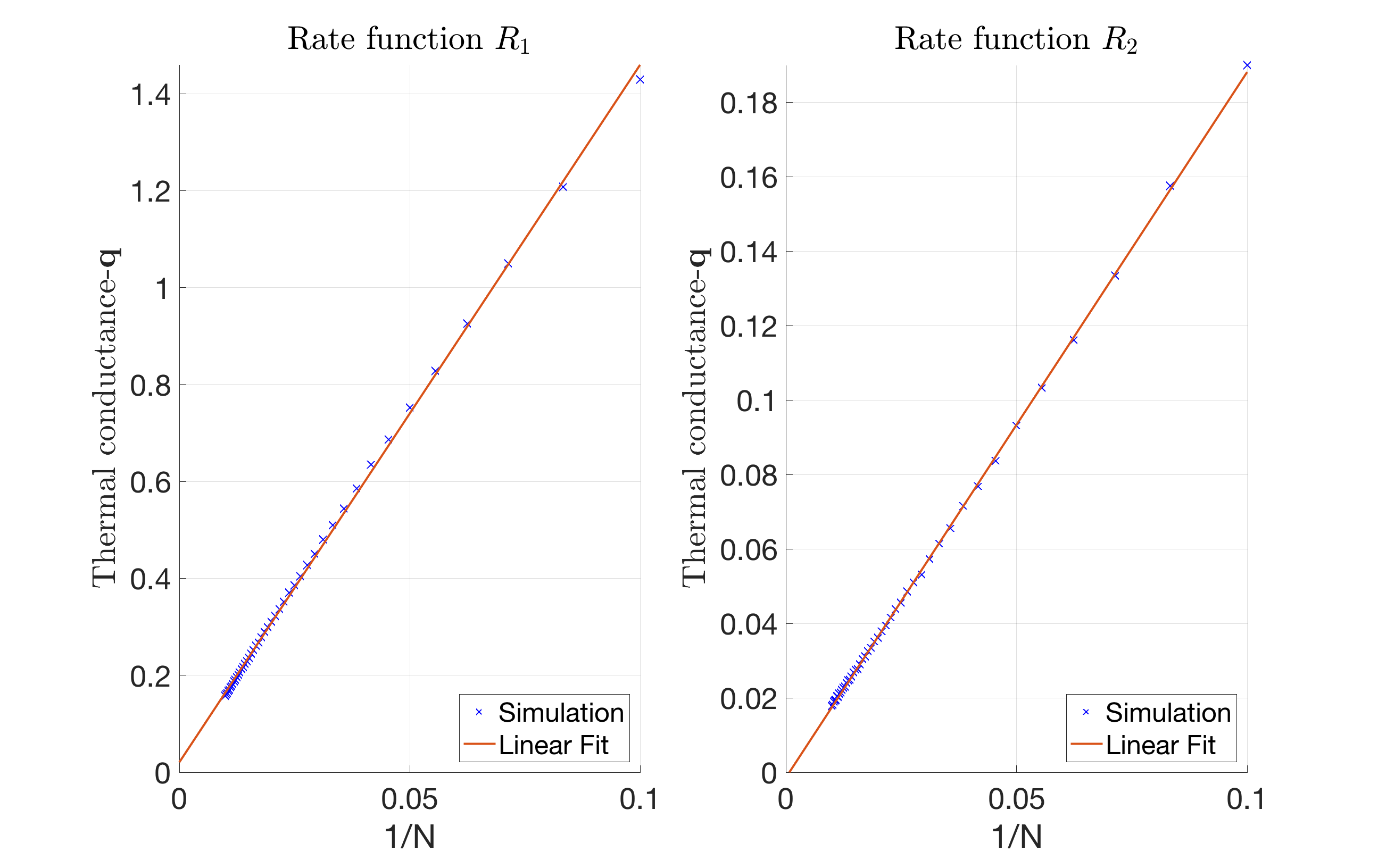}}
\caption{Thermal conductance for $\Phi^{1}_{t}$ and
  $\Phi^{2}_{t}$. Red line is the linear fit of $\kappa$ versus
  $1/N$. Left: $R_{1} = \sqrt{E_{1} + E_{2}}$. Right: $R_{2} =
  \sqrt{E_{1}E_{2}/(E_{1} + E_{2})}$.}
\label{fig1}
\end{figure}

According to the plot given above, we see that $\mathbf{q}$
is proportional to $1/N$ for both cases, although the thermal
conductance of $\Phi^{2}_{t}$ is much lower. In other words, in spite
of a much slower ergodicity and the presence of the  low energy site effect, $\Phi^{2}_{t}$ still gives a``normal''
thermal conductivity that is independent of the system size. The
effect of low energy site will 
quantitatively reduce the thermal conductivity, but not qualitatively
change the scaling of the thermal conductivity. In contrast, note that
many harmonic chains and anharmonic chains admit ``abnormal'' thermal
conductivities. We refer the review article \cite{lepri2003thermal} for
a summary of these results.

\subsection{Thermal conductivity of 2D model}

The thermal conductivity of a 2D stochastic energy exchange model is
also interesting. Obviously the low thermal conductivity of
$\Phi^{2}_{t}$ is mainly contributed by the occurence of low energy sites. The
occasional occurrence of a low energy site can block the energy transport for a
long time, and significantly reduce the thermal conductivity. This
problem can be alleviated by increasing the dimension of the
system. Instead of an 1D chain, we consider a 2D array of energy
sites. The upper and lower edges are adiabatic, while the left and
right edges connects to the heat bath. 

More precisely, we consider an $M\times N$ array of sites. An
exponential clock with rate $R = R_{1}$ or $R_{2}$ is associated to
each pair of nearest neighbor sites. When the clock rings, the rule of
energy redistribution is same as described in equation \eqref{update}. In addition,
sites with indices $(i,1)$ (resp. $(i,N)$) for $i = 1, \cdots, M$ are
connected to the left (resp. right) heat bath. The rule of energy
exchange with heat bath is same as in equation \eqref{updatebc}. 

The thermal conductivity $\kappa$ can then be defined and computed
analogously. We have
\begin{align}
\label{kappa2d}
  \kappa &= \frac{1}{M} \frac{1}{T_{R} - T_{L}}\int \left \{
             (\sum_{i = 1}^{M}\sum_{ j = 1}^{N-1}
  R(E_{i, j}, E_{i, j+1})\int_{0}^{1}p(E_{i, j} + E_{i, j+1})
             \mathrm{d}p - E_{i, j} ) \right .\\\nonumber
&+ \sum_{i = 1}^{M}R(T_{L}, E_{i,1})(E_{i,1} -
  \int_{0}^{\infty} \int_{0}^{1}p(E_{i,1} + x)e^{-x/T_{L}} \mathrm{d}p
  \mathrm{d}x )\\\nonumber
& \left .+ \sum_{i = 1}^{M}R(E_{i,N}, T_{R})(  \int_{0}^{\infty} \int_{0}^{1}p(E_{i,N} + x)e^{-x/T_{R}} \mathrm{d}p
  \mathrm{d}x - E_{i,N})\right \}
  \pi( \mathrm{d} \mathbf{E}) \,.
\end{align}
And again, we denote $\mathbf{q} = \kappa/(N+1)$ as the thermal conductance.

Similar as in the previous subsection, $\kappa$ is a computable
quantity. Let $t_{0}< t_{1} < t_{2} <\cdots$ be the time at which an
``horizontal'' energy
exchange, i.e., energy exchange between $E_{i,j}$ and $E_{i, j \pm 1}$
(or heat bath) occurs. Let $J_{k}$ be the energy flux from right to left
associated to the energy exchange event occurring at time $t_{k}$. If
the energy exchange event is between site $(i,j)$ and site $(i,j+1)$, we have
$J_{k} = E_{i,j}(t_{k}^{+}) - E_{i,j}(t_{k})$. If the energy exchange is
between site $1$ (resp. site $N$) and the left (resp. right) boundary,
we have $J_{k} = E_{i,1}(t_{k}) - E_{i,1}(t_{k}^{+})$ (resp. $J_{i} =
E_{i,N}(t_{k}^{+}) - E_{i,N}(t_{k})$).

A similar approach as in Theorem \ref{lln} implies
\begin{equation}
  \label{kappa2d2}
  \kappa = \lim_{T \rightarrow \infty} \frac{1}{M}\frac{1}{T_{R}
    - T_{L}} \sum_{t_{k} < T} J_{k} \,.
\end{equation}

{\bf Numerical Simulation 2:} Let $T_{L} = 1$, $T_{R} = 2$, $N = 50$,
and $M = 1, 2, \cdots, 20$. By equation \eqref{kappa2d2}, we can
simulate the thermal conductance $\mathbf{q}$ by computing 
$$
 \lim_{T\rightarrow \infty} \frac{1}{M}\frac{1}{T_{R}
    - T_{L}} \sum_{t_{k} < T} J_{k}
$$
over a long trajectory. In our simulation $T$ is chosen to be $1
\times 10^{7}$. Then we compare $\mathbf{q}$ for each $M$ from $1$ to
$20$. Simulation results of $\mathbf{q}$ vs. $M$ for $\Phi^{1}_{t}$ and $\Phi^{2}_{t}$ are presented in Figure \ref{fig2}.

\begin{figure}[h]
\centerline{\includegraphics[width = \linewidth]{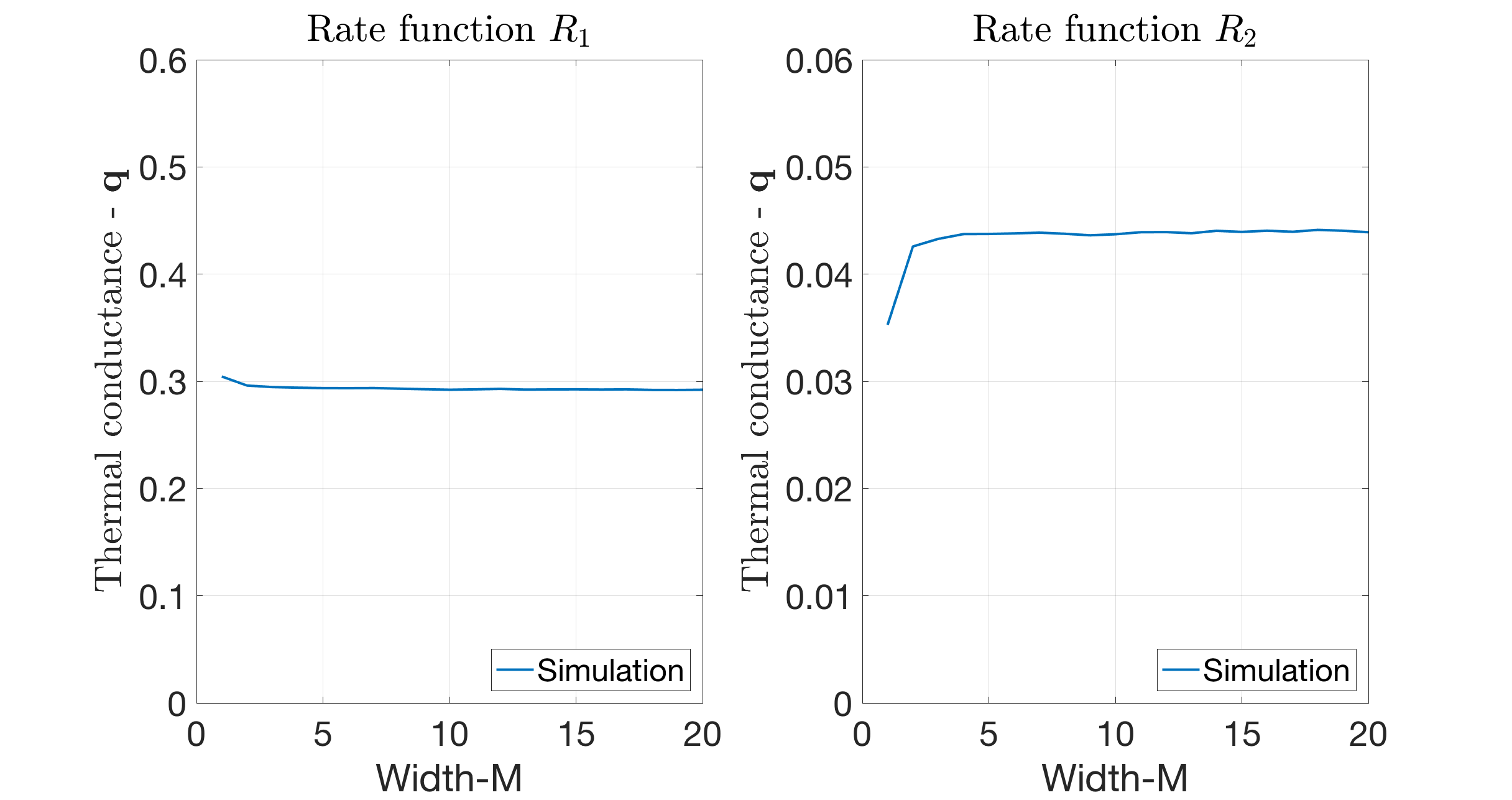}}
\caption{Thermal conductance $\mathbf{q}$ vs. $M$ of the 2D system for $\Phi^{1}_{t}$ and
  $\Phi^{2}_{t}$. Length of the chain is fixed as $N = 50$. $M$
  varies from $1$ to $20$. Left: $R_{1} = \sqrt{E_{1} + E_{2}}$. Right: $R_{2} =
  \sqrt{E_{1}E_{2}/(E_{1} + E_{2})}$.}
\label{fig2}
\end{figure}

Figure \ref{fig2} confirms our speculation. With rate function $R_{1}
= \sqrt{E_{1} + E_{2}}$, the thermal conductivity changes
inconspicuously even the width of the system increases. With rate function $R_{2} =
  \sqrt{E_{1}E_{2}/(E_{1} + E_{2})}$, $\mathbf{q}$ (as well as $\kappa$) increases significantly
  when $M$ changes from $1$ to $2$, and keeps increasing with
  increasing $M$. This demonstrates the dimension
  effect. When a site loses most of its energy in an energy exchange
  and becomes ``silent'' for a while, the energy transport is
  completely blocked in an 1D model. With an extra dimension, the
  energy can still be transported by circumventing the ``silent''
  site. In addition, the probability that the energy transport is
  completed blocked becomes much lower in a 2D model.

\section{Comparison of local thermodynamic equilibrium(LTE)}
The local thermodynamical equilibrium (LTE) assumption means that
although the entire system is nonequilibrium, the marginal distribution of the steady state with
respect to a ``local'' subset is still close to a thermal
equilibrium. The existence of LTE is equivalent to the existence of a
well-defined local temperature. In the study of microscopic heat
conduction models, the existence of LTE usually means the marginal
distribution of NESS with respect to finite many local sites converges
to a thermal equilibrium as the length of the chain goes to infinity. We refer 
\cite{li2016local, ravishankar2007local} for further discussion and
known rigorous results about the existence of LTE.

\subsection{Nonequilibrium steady state under the LTE assumption}

The two rate functions in Section 2 are chosen in a way that the theoretical
thermal equilibrium can be explicitly given. We start this subsection with the
following Proposition.

\begin{prop}
\label{infinite}
Assume the chain is infinitely long on both sides. The process $\Phi^{1}_{t}$
(resp. $\Phi^{2}_{t}$) admits a family of invariant probability
measures 
\begin{equation}
  \label{invinfty}
  \pi = \prod_{i = -\infty}^{\infty} \pi_{i}^{T} \quad T 
> 0 \,,
\end{equation}
where $\pi_{i}^{T}$ are i.i.d. exponential random variables
(resp. Gamma random variables) with mean $T$ (resp. parameters $1/2$ and
$T$). 
\end{prop}
\begin{proof}
Since there is no energy exchange with the boundary, it is sufficient
to check the interaction between $E_{i}$ and $E_{i+1}$. 

When starting from the probability distribution $\pi$ given in the theorem, the probability that
$\Phi^{1}_{t}$ leaves $(E_{i}, E_{i+1})$ on the next infinitesimal
time interval $(0, \mathrm{d}t) $ is
$$
  \sqrt{E_{i} + E_{i+1}} e^{-E_{i}/T} e^{-E_{i+1}/T} \mathrm{d}t \,.
$$
On the other hand, when starting from $\pi$, the probability density
that $\Phi^{1}_{t}$ enters an infinitesimal neighborhood of $(E_{i},
E_{i+1})$ on the same time interval is
\begin{align}
\label{phi1}
& \mathrm{d}t \cdot\int_{0}^{E_{i} + E_{i+1}} \frac{1}{E_{i} + E_{i+1}} \sqrt{x +
    (E_{i} + E_{i+1} - x)} e^{-x/T} e^{-(E_{i} + E_{i+1} - x)/T}
  \mathrm{d}x  \\\nonumber
=& \sqrt{E_{i} + E_{i+1}} e^{-E_{i}/T} e^{-E_{i+1}/T} \mathrm{d}t\,.
\end{align}
Therefore, $\pi$ is invariant for $\Phi^{1}_{t}$ if $\pi_{i}$ are
i.i.d. exponential distributions. 

The case of $\Phi^{2}_{t}$ is the same. The probability that
$\Phi^{2}_{t}$ leaves $(E_{i}, E_{i+1})$
is
\begin{align}
\label{phi2}
  &\sqrt{\frac{E_{i}E_{i+1}}{E_{i} + E_{i+1}}}\cdot
  \frac{1}{\sqrt{T} \Gamma(1/2)}\frac{1}{\sqrt{E_{i}}}e^{-E_{i}/T} \cdot \frac{1}{\sqrt{T} \Gamma(1/2)}\frac{1}{\sqrt{E_{i+1}}}
  e^{-E_{i+1}/T} \mathrm{d}t \\\nonumber
=&\frac{1}{T \pi} \frac{1}{\sqrt{E_{i} +
      E_{i+1}}} e^{-(E_{i} + E_{i+1})/T} \mathrm{d}t \,. 
\end{align}
The probability density that $\Phi^{2}_{t}$ enters an infinitesimal neighborhood of $(E_{i},
E_{i+1})$ is
\begin{align}
\label{eq10}
&  \mathrm{d}t \cdot\int_{0}^{E_{i} + E_{i+1}} \frac{1}{E_{i} +
   E_{i+1}} \sqrt{\frac{x(E_{i} + E_{i+1} - x)}{x +
    (E_{i} + E_{i+1} - x)}} \cdot \frac{1}{\sqrt{T}
   \Gamma(1/2)}\frac{1}{\sqrt{x}} e^{-x/T}\\\nonumber
& \cdot \frac{1}{\sqrt{T}
   \Gamma(1/2)}\frac{1}{\sqrt{E_{i} + E_{i+1} - x}} e^{-(E_{i} + E_{i+1} - x)/T}
  \mathrm{d}x  \\\nonumber
=& \frac{1}{T \pi} \frac{1}{\sqrt{E_{i} +
      E_{i+1}}} e^{-(E_{i} + E_{i+1})/T} \mathrm{d}t \,.
\end{align}
Therefore, $\pi$ is invariant for $\Phi^{2}_{t}$ if $\pi_{i}$ are
i.i.d. Gamma distributions. This completes the proof.
\end{proof}

This theoretical thermal equilibrium does not work well for a
finite chain due to boundary effects. However, we are curious about whether the
marginal distribution of the NESS with respect to finite local sites
converges to i.i.d. exponential (or Gamma) distributions when $N
\rightarrow \infty$. If the answer is yes, then the LTE is
established. Note that here we adopt a strict definition of the
LTE. We see that LTE is achieved only if the marginal distribution of the NESS with
respect to {\it many} local sites converges to the thermal equilibrium
described before. There are also literatures about weaker versions of
LTE \cite{mejia2001coupled}, i.e., the marginal distribution with respect to {\it one} site or
one point. 

We plan to use the energy profile and the thermal conductivity to
preliminarily check whether LTE is achieved. It is not difficult to calculate the
theoretical energy flux $J_{i, i+1}$ if we know the joint distribution of $(E_{i},
E_{i+1})$. This theoretical flux can be then used to compute a
theoretical energy profile. Then we can compare the theoretical energy
profile and its empirical counterpart (which is relatively easy to compute).

The following two propositions follow from simple calculations. 

\begin{prop}
\label{flux1}
If the joint marginal distribution of the invariant probability
measure of $\Phi^{1}_{t}$ with respect to site $E_{i}$ and $E_{i+1}$
is the product measure of two exponential distributions with mean $T$
and $\hat{T}$ respectively, then the mean energy flux from site $i+1$ to
site $i$ is
\begin{equation}
  \label{fluxphi1}
  J_{i,i+1} = \frac{\sqrt{\pi} (3T^{2} + 9 T^{3/2} \hat{T}^{1/2} + 11 T
    \hat{T} + 9 T^{1/2} \hat{T}^{3/2} +  3\hat{T}^{2})}{8 (T^{1/2} +
    \hat{T}^{1/2})^{3}} (\hat{T} - T)\,.
\end{equation}
\end{prop}
\begin{proof}
We have
\begin{equation}
  \label{eq11}
  J_{i, i+1} = \int_{0}^{\infty}\int_{0}^{\infty} \frac{y-x}{2}\cdot
  \sqrt{x + y}
  \cdot\frac{1}{T}e^{-x/T}\cdot\frac{1}{\hat{T}}e^{-y/\hat{T}}
  \mathrm{d}x \mathrm{d}y \,.
\end{equation}
Let $ u = y - x$ and $v = x+y$. The rest are straightforward
calculation about the integral. 
\end{proof}

\begin{prop}
\label{flux2}
If the joint marginal distribution of the invariant probability
measure of $\Phi^{2}_{t}$ with respect to site $E_{i}$ and $E_{i+1}$
is the product measure of two Gamma distributions with parameters
$(1/2, T)$ and $(1/2, \hat{T})$ respectively, then the mean energy flux from site $i+1$ to
site $i$ is
\begin{equation}
  \label{fluxphi2}
  J_{i,i+1} = \frac{T^{1/2} \hat{T}^{1/2}(T + 3T^{1/2} \hat{T}^{1/2} + \hat{T})}{4
    \sqrt{\pi}(T^{1/2} + \hat{T}^{1/2})^{3}}  (\hat{T} - T)\,.
\end{equation}

\end{prop}
\begin{proof}
We have
\begin{equation}
  \label{eq12}
    J_{i, i+1} = \int_{0}^{\infty}\int_{0}^{\infty} \frac{y-x}{2}\cdot
  \sqrt{\frac{xy}{x + y}}
  \cdot\frac{1}{\sqrt{T}\Gamma(1/2)}e^{-x/T}\cdot\frac{1}{\sqrt{\hat{T}}\Gamma(1/2)}e^{-y/\hat{T}}
  \mathrm{d}x \mathrm{d}y \,.
\end{equation}
Let $ u = y - x$ and $v = x+y$. The rest are straightforward
calculation about the integral. 
\end{proof}

\subsection{Numerical study of the LTE assumption}
We propose the following three numerical simulations to check the validity
of the LTE assumption. Note that the rule of boundary interaction is
different from that in the middle of the chain. Hence marginal
distributions with respect to boundary sites always have boundary
effects, and the system does not reach thermal equilibrium even if
$T_{L} = T_{R}$. We will show that when the length of the chain
increases, the boundary effect gradually disappears. At the
limit, the marginal distribution with respect to finite many sites
that are in the middle of the chain converges to the theoretical
thermal equilibrium given by Proposition \ref{infinite}.

{\bf Numerical Simulation 3.} We first simulate the marginal
distribution with respect to a {\it single} site. Let $T_{L} = 1$,
$T_{R} = 2$, $N = 10, 20, 40, 60, 80$. We simulate processes
$\Phi^{1}_{t}$ and $\Phi^{2}_{t}$ over a long trajectory and collect
the energy profile at sampling times $h, 2h, \cdots, 8\times
10^{6}h$. $h$ is chosen to be $2$ for $\Phi^{1}_{t}$ and $10$ for
$\Phi^{2}_{t}$. The time-$h$ skeleton of a time-continuous
Markov process preserves its invariant probability measure. Hence we can
compute the marginal distribution of the invariant probability measure
with respect to each site. 

Then we compare the sample with respect to each site with a desired Gamma
distribution. This step is done by using the {\it gamfit} function in
MATLAB. Parameters $(\alpha, \beta)$ of the Gamma distribution with
respect to each set are demonstrated in 
Figure \ref{fig3a} and Figure \ref{fig3b}. We can find that the
marginal distribution with respect to a non-boundary site is
approximately a Gamma distribution with parameters $(1, \beta_{i})$
for $\Phi^{1}_{t}$ and $(1/2, \beta_{i})$ for $\Phi^{2}_{t}$, where
$\beta_{i}$ changes with the site index. Note that an exponential
distribution with mean $\lambda$ is a Gamma distribution with
parameters $(1, \lambda^{-1})$. Hence our numerical result is consistent with the
marginal distribution obtained in Proposition \ref{infinite}. 

\begin{figure}[h]
\centerline{\includegraphics[width = \linewidth]{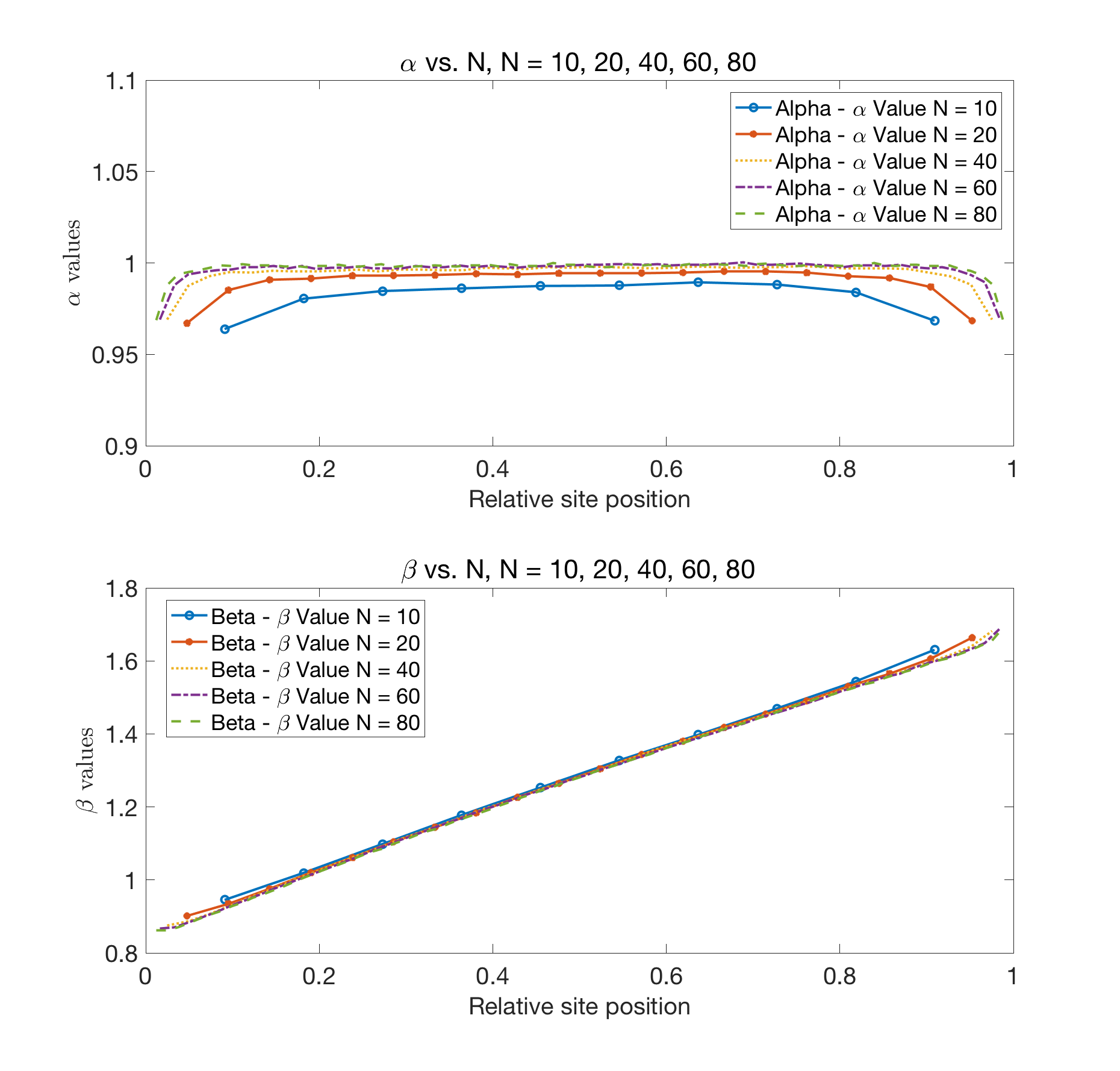}}
\caption{Parameters $\alpha$ and $\beta$ of the Gamma distribution
  fitted from marginal distributions with respect to all sites. The
  length of the chain is $N = 10, 20, 40, 60, 80$. Rate function
  $R_{1} = \sqrt{E_{1} + E_{2}}$. }
\label{fig3a} 
\end{figure}
 
 \begin{figure}[h]
\centerline{\includegraphics[width = \linewidth]{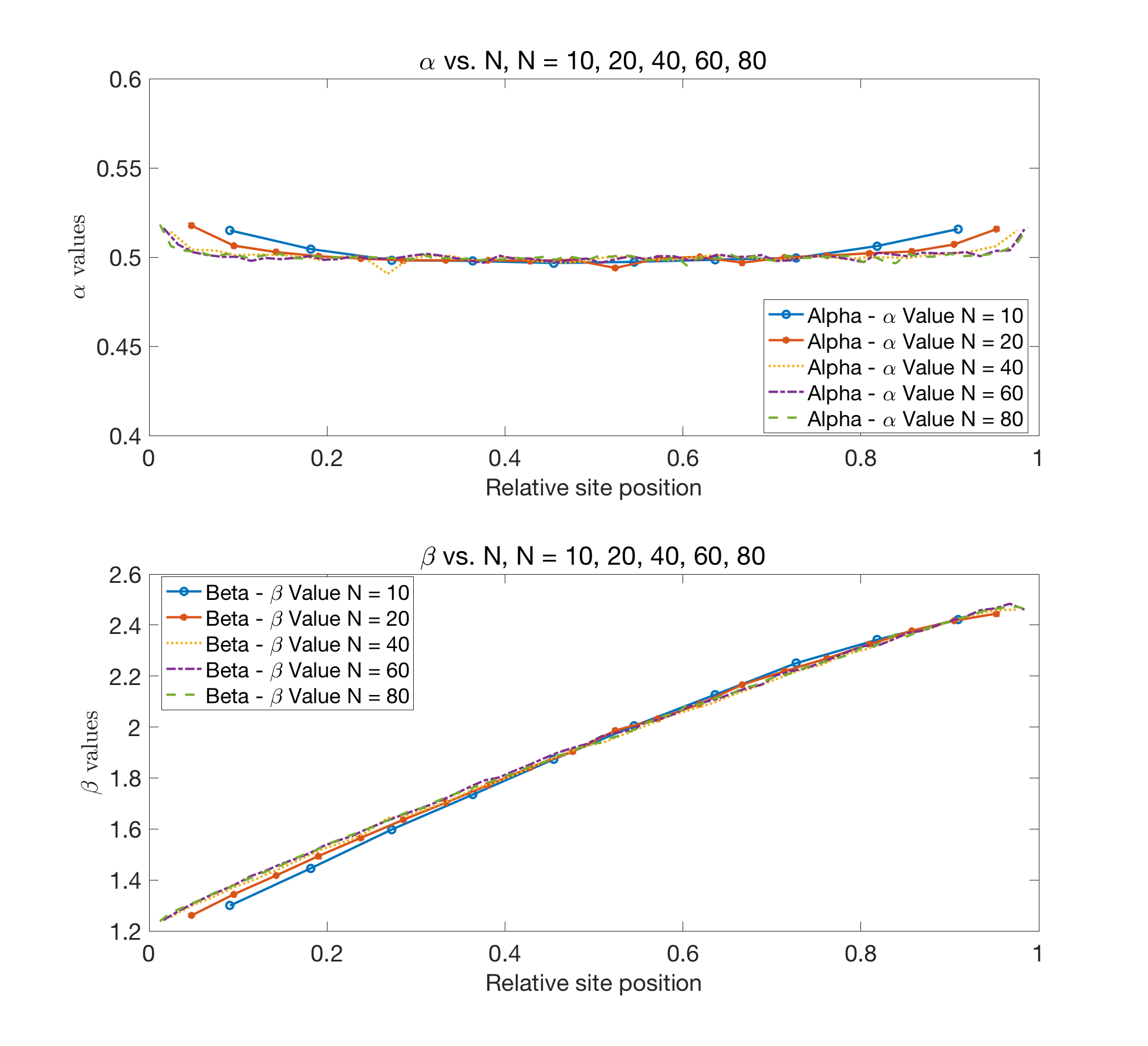}}
\caption{Parameters $\alpha$ and $\beta$ of the Gamma distribution
  fitted from marginal distributions with respect to all sites. The
  length of the chain is $N = 10, 20, 40, 60, 80$. Rate
  function $R_{2} = \sqrt{E_{1}E_{2}/(E_{1} + E_{2})}$.}
\label{fig3b} 
\end{figure}

The goodness of the fit is done by a Chi-square test. We
divide the domain into $31$ bins $[0, 0.2), \cdots, [5.8, 6.0), [6.0,
\infty)$. Let $p_{i}$ be the theoretical probability of the desired
Gamma distribution at each interval and $n_{i}$ be the number of
samples falling to this interval. We calculate
\begin{equation}
  \label{chisquare}
  \chi^{2} = \sum_{ i = 1}^{31}\frac{(n_{i} - \mathbf{N}
    p_{i})^{2}}{\mathbf{N} p_{i}}
\end{equation}
for each site. If the marginal distribution satisfies a Gamma
distribution, the $\chi^{2}$ test statistics should be smaller than
the $95$th percentile of a $\chi^{2}$ distribution with $30$ degrees of
freedom. Figure \ref{fig4a} and Figure \ref{fig4b} shows our result for the goodness of the
fit. We can see that when the chain is long enough, the marginal
distribution with respect to a non-boundary site is very closed to a Gamma
distribution in both cases, which is exactly the thermal equilibrium
we found in Proposition \ref{infinite}. In other words, LTE is achieved for a {\it
  single} site in the chain as the length of the chain grows.

\begin{figure}[h]
\centerline{\includegraphics[width = \linewidth]{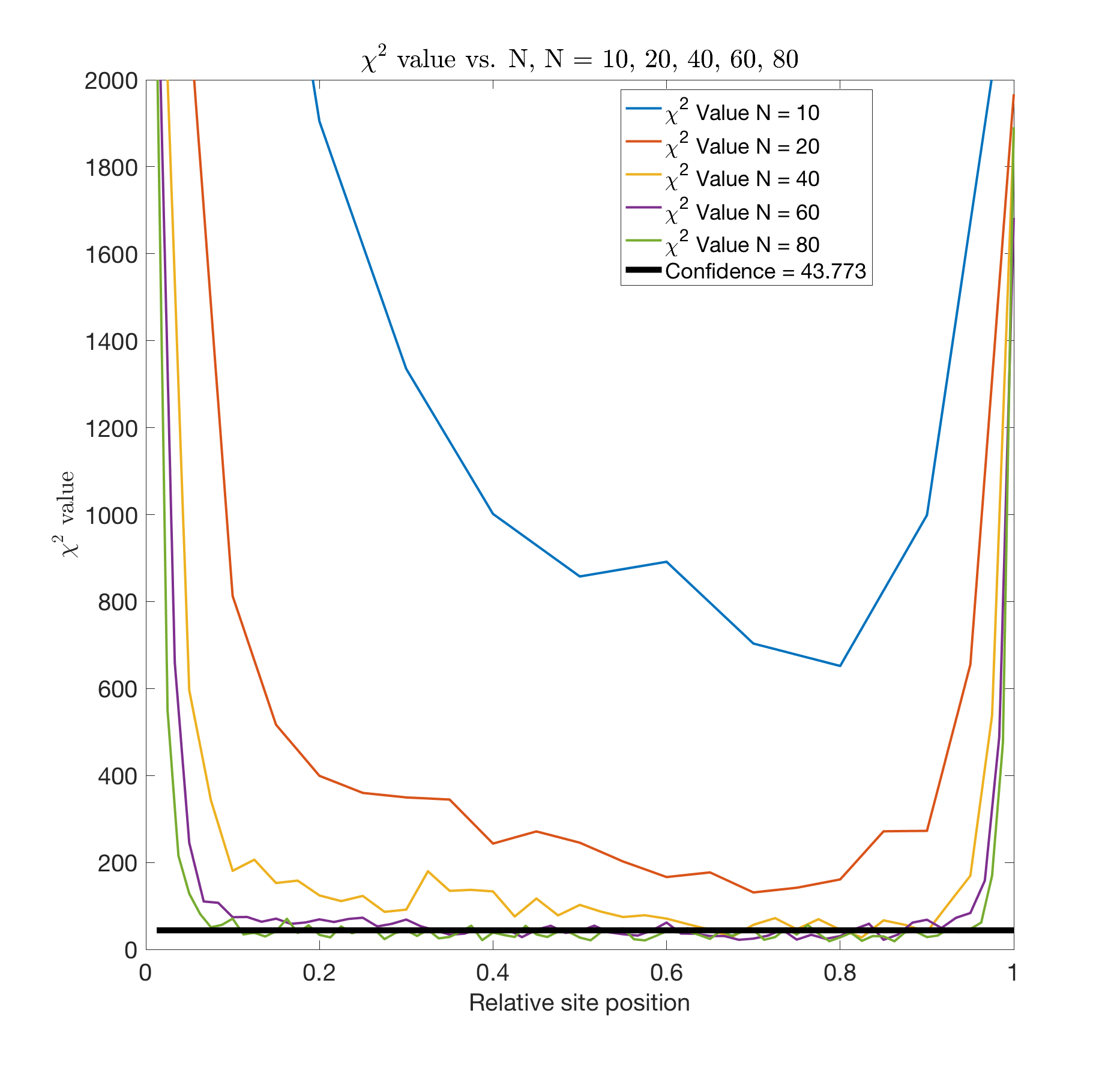}}
\caption{Values of $\chi^{2}$-test statistics (equation \eqref{chisquare}) of the marginal
  distribution of each site. $x$-axis: site index. $y$-axis: $\chi^{2} $. Rate
  function $R_{1} = \sqrt{E_{1} + E_{2}}$.}
\label{fig4a} 
\end{figure}
 
 \begin{figure}[h]
\centerline{\includegraphics[width = \linewidth]{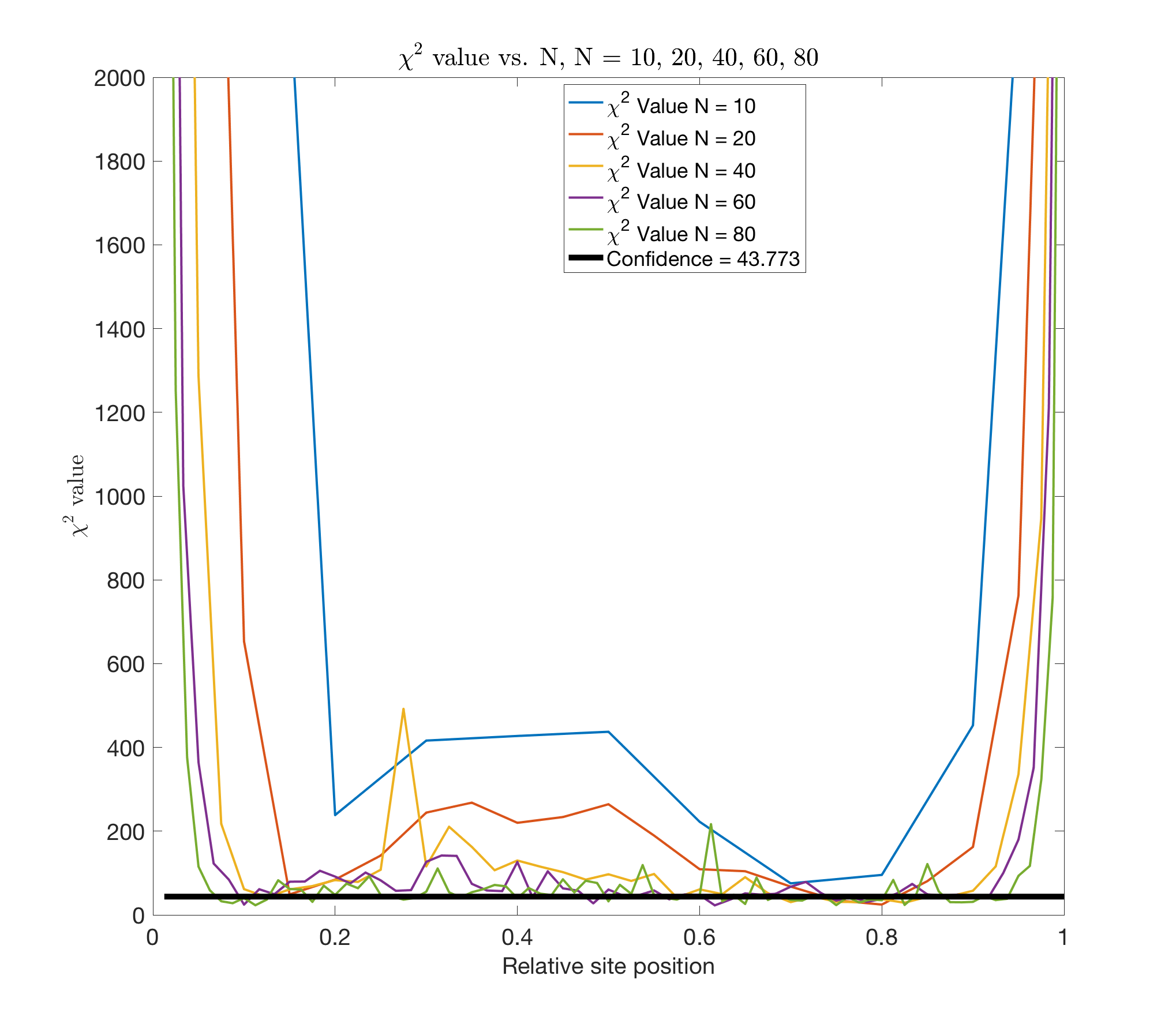}}
\caption{Values of $\chi^{2}$ test statistics (equation \eqref{chisquare}) of the marginal
  distribution of each site. $x$-axis: site index. $y$-axis: $\chi^{2} $. Rate
  function $R_{2} = \sqrt{E_{1}E_{2}/(E_{1} + E_{2})}$.}
\label{fig4b} 
\end{figure}

\medskip

{\bf Numerical Simulation 4.} The next task is to use the energy
profile to check the LTE assumption. Assume LTE is achieved, the
theoretical energy flux can be obtained from equations \eqref{fluxphi1}
and \eqref{fluxphi2}. Then we can compare the empirical energy profile with the predicted ones when assuming the LTE. The
result of {\bf Numerical Simulation 3} shows that at the boundary the
marginal distribution is far from the Gamma distribution. Hence we can
only use Propositions \ref{flux1} and \ref{flux2} to predict the
energy profile in the middle. The
predicted energy profile under the LTE assumption is obtained in the 
following way. Assume that $E_{5}$ and $E_{36}$ are equal to those in the
empirical energy profile. Since the mean
energy flux $J_{i,i+1}$ is independent of the choice of $i$, we can solve a
nonlinear equation involving $E_{6}, E_{7}, \cdots, E_{35}$
numerically such that $J_{6, 7} = J_{7, 8} = \cdots = J_{34,
  35}$, where terms $J_{i,i+1}$ are from equation
\eqref{fluxphi1} for $\Phi^{1}_{t}$ and \eqref{fluxphi2} for $\Phi^{2}_{t}$. This gives the predicted energy profile from site $6$ to site
$35$. The predicted and empirical energy profiles for $\Phi^{1}_{t}$
and $\Phi^{2}_{t}$ are compared in Figure \ref{fig5}. We can
find that in both cases the predicted energy profile is very close to
the empirical one. Note that in the energy profile, the mean energy of
the left (resp. right) boundary site is not close to $T_{L}$
(resp. $T_{R}$). This is because the rule of boundary
interaction is different from that of non-boundary sites. In
particular, the rate of an energy exchange with the left (resp. right)
boundary is $R(T_{L}, E_{1})$ (resp. $R(E_{N}, T_{R})$) regardless the
amount of energy drawn at the boundary. Hence the
effective temperature ``felt'' by a boundary site is not the heat
bath temperature. 

 \begin{figure}[h]
\centerline{\includegraphics[width = \linewidth]{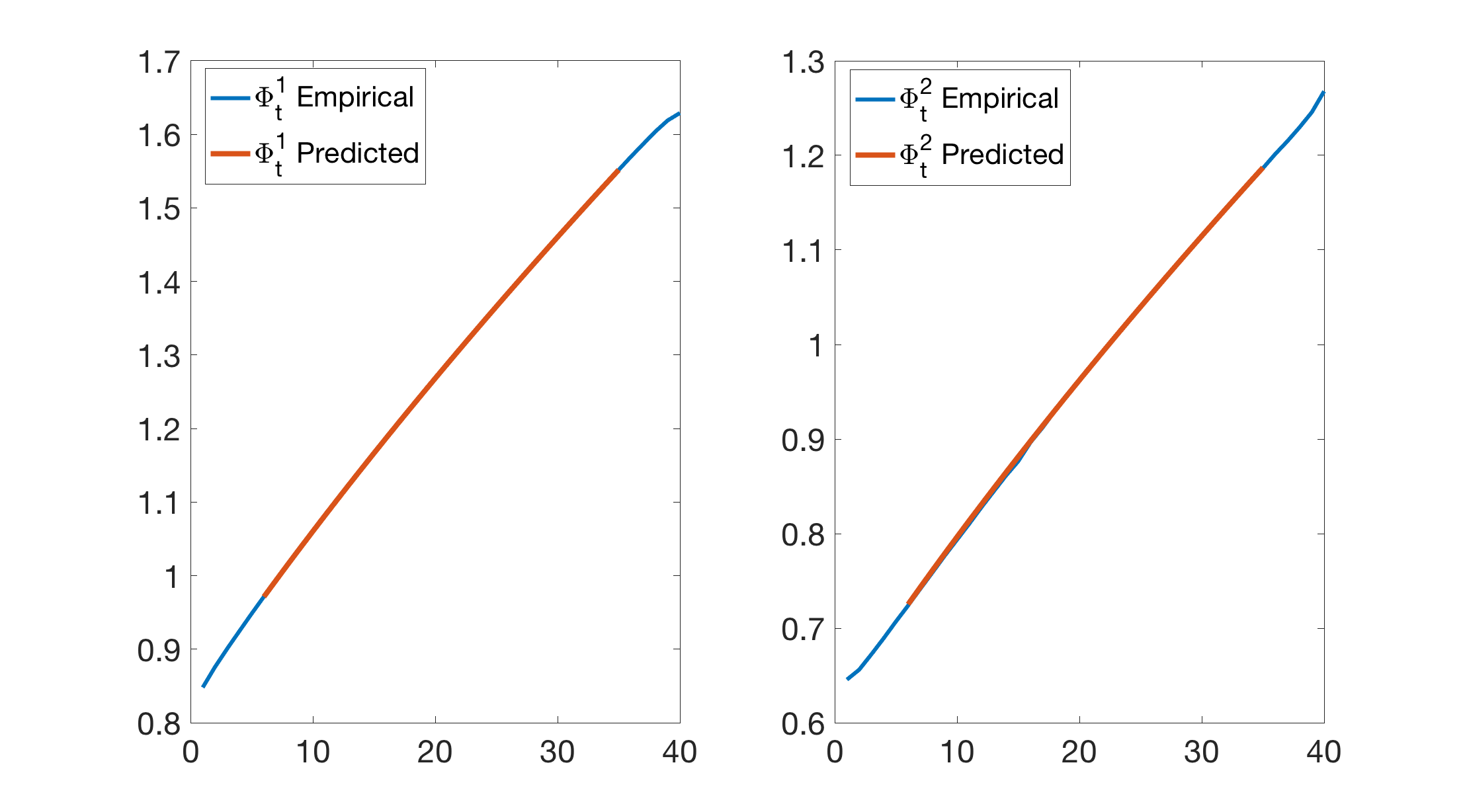}}
\caption{A comparison of predicted and empirical energy profiles. Left
panel: $R_{1} = \sqrt{E_{1} + E_{2}}$. Right panel: $R_{2} =
\sqrt{E_{1}E_{2}/(E_{1} + E_{2})}$. }
\label{fig5} 
\end{figure}

The result in Figure \ref{fig5} suggests that equations
\eqref{fluxphi1} and \eqref{fluxphi2} can produce good approximations
of the macroscopic energy profile. However, the energy profile
alone is not sufficient for us to claim the
existence of LTE, as dependent marginal distributions can still
produce the same energy profile. In fact, when the chain is not long enough, the theoretical
mean energy flux by assuming LTE is quite different from the empirical
energy flux. In order to check whether LTE is achieved, we need to accurately compute the marginal
distribution with respect to nearest neighbor sites.

\medskip

{\bf Numerical Simulation 5.} Finally, we simulate the joint marginal
distribution with respect to two nearest neighbor sites at the center
of the chain. We simulate $8$ long trajectories to generate the joint marginal distribution with
respect to $(E_{N/2}, E_{N/2 + 1})$ for increasing $N$. For each trajectory, we collect
samples of $(E_{N/2}, E_{N/2 + 1})$ at each sampling time $h, 2h,
\cdots, 2\times 10^{8} h$. We choose $h = 0.25$ for $\Phi^{1}_{t}$ and
$h = 1$ for $\Phi^{2}_{t}$ because the average clock rate of
$\Phi^{2}_{t}$ is lower. The reason of doing this is because the time-$h$ skeleton of
a time-continuous Markov process preserves its invariant probability
measure. This approach gives us $1.6 \times 10^{9}$ samples. We need
these many samples to achieve the accuracy needed for verifying the existence of LTE.

Let $a_{0}, a_{1}, \cdots, a_{16} = 0, 0.1,\cdots, 1.6$ and $a_{17} = \infty$. We define two
auxiliary random variables $Z_{1}$ and $Z_{2}$ that represent the
discretization of $E_{N/2}$ and $E_{N/2 + 1}$ with respect to the
partition generated by $a_{0}, a_{1}, \cdots, a_{17}$, respectively. $Z_{1} = i$ (resp. $Z_{2}
= i$) if and only if $E_{N/2} \in [a_{i-1}, a_{i})$ (resp. $E_{N/2 +
  1} \in [a_{i-1}, a_{i})$). In other words
$Z_{1}, Z_{2}$ takes the value $1$ to $17$. We use the collected
$\mathbf{N} = 1.6 \times 10^{9}$ samples to estimate the
probability distributions of $Z_{1}, Z_{2}$ as well as their joint
distributions. If $Z_{1}, Z_{2}$ converges to two independent random
variables as $N \rightarrow \infty$, we believe this implies the
independence of $E_{N/2}$ and $E_{N/2 + 1}$. 

Then we use the extrapolation of $\chi^{2}$ values to decide 
whether $Z_{1}$ and $Z_{2}$ are independent. For $i, j = 1, \cdots,
17$, define $O_{i}$, $O_{j}$, and $O_{ij}$ be the sample size
corresponding to
$\{Z_{1} = i\}$, $\{Z_{2} = j\}$, and $\{Z_{1} = i, Z_{2} =
j\}$ respectively. Let $E_{ij} = O_{i}O_{j}/\mathbf{N}$ be the
expected count of $O_{ij}$. The $\chi^{2}$-value is given by
\begin{equation}
  \label{chi2error}
  \chi^{2}_{N} = \sum_{i = 1}^{17}\sum_{j = 1}^{17}\frac{(O_{ij} -
    E_{ij})^{2}}{E_{ij}} \,.
\end{equation}
If $Z_{1}$ and $Z_{2}$ are independent, then $\chi^{2}_{N}$ should
be less than the $95$th percentile of a chi-square distribution with
degree of freedom $16 \times 16$, denoted by $p_{95}$. We compute $\chi^{2}_{N}$ up to
$N = 160$ for $\Phi^{1}_{t}$ and $N = 240$ for $\Phi^{2}_{2}$. (Because
the simulation of $\Phi^{1}_{t}$ is slower.) Then we plot $N^{-1}$ versus $\sqrt{\chi^{2}_{N}}$, use a
linear extrapolation to estimate the chi-square score of the limit
case when $N \rightarrow \infty$, and compare it with the square root
of $p_{95}$. The result is demonstrated in
Figure \ref{margin}. The linear extrapolation shows that
$\chi^{2}_{\infty}$ values are less than $p_{95}$ for both
$\Phi^{1}_{t}$ and $\Phi^{2}_{t}$. Hence we believe this simulation result gives a convincing evidence that when $N
\rightarrow \infty$, the marginal distribution with respect to the
nearest neighbor sites in the middle of the chain converges to independent random
variables. 

\begin{figure}[h]
\centerline{\includegraphics[width = \linewidth]{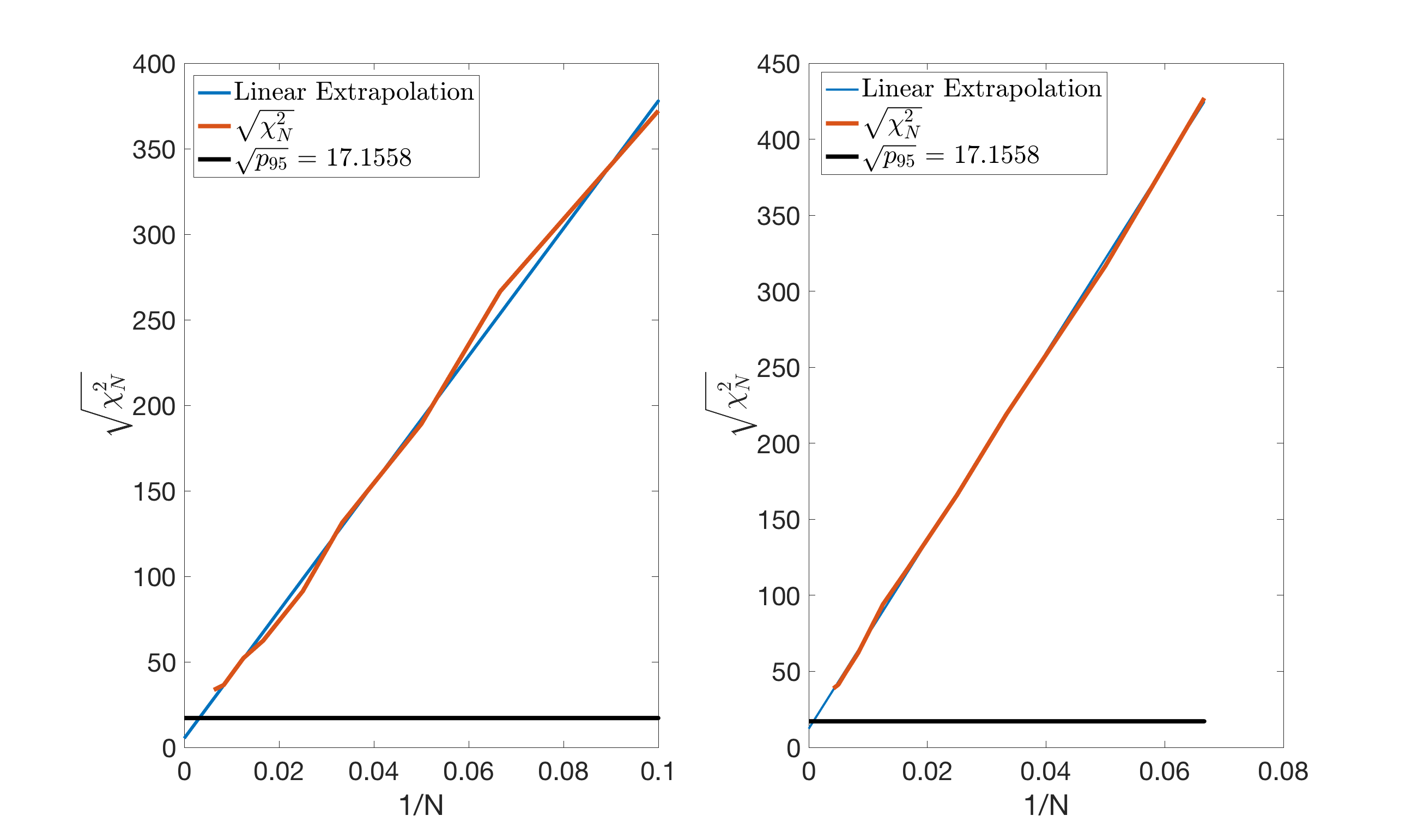}}
\caption{Left: Extrapolation of $\chi^{2}_{N}$ for
  $\Phi^{1}_{t}$. Right: Extrapolation of $\chi^{2}_{N}$ for
  $\Phi^{2}_{t}$. Red: $\sqrt{\chi^{2}_{N}}$ vs. $1/N$. Blue: linear
  extrapolation. Black: $\sqrt{p_{95}}$, square root of the $95$th
  percentile of the $\chi^{2}$ distribution with $256$ degrees of freedom.}
\label{margin}
\end{figure}

We remark that in order to get the desired accuracy, one needs to
efficiently generate large number of samples to approximate the
invariant probability measure. This is achieved by using
the Hashing-Leaping Method proposed in \cite{li2015fast} and parallel
computation.  

\section{Conclusion}
In this paper we study a stochastic energy exchange model with two
different rate functions that corresponding to different ways of
reduction from a deterministic billiards-like heat conduction
model. The Markov chain generated by the model with rate function
$R_{1} = \sqrt{E_{1} + E_{2}}$ (resp. $R_{2} = \sqrt{E_{1}E_{2}/(E_{1}
+ E_{2})}$) is denoted by $\Phi^{1}_{t}$ (resp. $\Phi^{2}_{t}$). Two
processes have fundamental difference when a low energy site appears. In
$\Phi^{1}_{t}$, a low energy site can be quickly ``rescued'' by its
neighbors. However, the rate function of $\Phi^{2}_{t}$ means a low
energy site can only be recovered by itself, which usually takes a
long time. It is known that this low energy site effect causes the
speed of convergence to the invariant probability measure much
slower. We are interested in the difference of macroscopic
thermodynamic properties caused by this difference. Since an explicit
formula of the nonequilibrium steady state (NESS) is not possible, we carry out a series of numerical studies in this paper.

The first study is about the thermal conductivity. We first proved
that the thermal conductivity is a well-defined and computable
quantity. Our simulations
find that both models have thermal conductances that are
proportional to $1/N$. This implies a ``normal'' thermal conductivity
that is independent of the system size. In other words, the low energy site effect of
$\Phi^{2}_{t}$ does not qualitatively change the
thermal conductivity. However, the thermal conductivity is
quantitatively reduced by the low energy site effect. This can be verified
by comparing the thermal conductivities of 1D and 2D models. In a 2D model
the energy transport can bypass the low energy site. Hence an increase of
thermal conductivity for $\Phi^{2}_{t}$ is observed in 2D, while the
thermal conductivity of $\Phi^{1}_{t}$ is roughly unchanged. 

The next study is on the marginal distributions at the NESS, which is
related to the existence of local thermodynamic equilibrium (LTE). Our
numerical and analytical studies show that when the chain is
sufficiently long, for both $\Phi^{1}_{t}$ and $\Phi^{2}_{t}$, the marginal distribution with respect to a
non-boundary site approaches to a theoretical thermal
equilibrium (Gamma distribution). Additional carefully designed numerical studies reveal that the
marginal distribution of nearest neighboring sites also approaches to
the product of two independent Gamma distributions regardless of the rate function. However, the low energy site effect of $\Phi^{2}_{t}$
causes the chain to be more ``sticky''. As a result, for the same
system size, nearest neighbor
sites of $\Phi^{2}_{t}$ are more dependent than those of
$\Phi^{1}_{t}$. 

Understanding the consequence of the ``low energy site effect'' is an
important step in the derivation of macroscopic thermodynamic laws from
nonequilibrium billiards-like dynamics. Recall that the stochastic
energy exchange model $\Phi^{2}_{t}$ serves as an approximation of the
time evolution of the energy profile of a billiard model. In our
recent paper \cite{li2018billiards}, we consider many particles that
are trapped in the same cell as described in Figure \ref{table}. A more realistic stochastic
energy exchange model is then derived from simulating this billiard
model. A weaker ``low energy site effect'' is still observed in this
stochastic energy exchange model. The new stochastic energy
exchange model in \cite{li2018billiards} is very important because it
has a mesoscopic limit equation. Many important macroscopic
thermodynamic properties like the Fourier's law, 
the long-range correlation, and the fluctuation theorem, can be
derived from this mesoscopic limit equation rigorously. Numerical
computation in this paper shows that the ``low energy site 
effect'' does not qualitatively change key macroscopic thermodynamic
properties. This not only answers questions about $\Phi^{2}_{t}$ asked by
several researchers in the field, but also makes the ongoing and
future study on the new stochastic energy exchange model in
\cite{li2018billiards} more convincing. In this sense, the study in
the present paper improves our understanding about how macroscopic thermodynamic laws are
derived from microscopic Hamiltonian dynamics.

\bibliography{myref}
\bibliographystyle{amsplain}
\end{document}